\documentclass{IEEEtran}
\makeatletter
\def\ps@headings{
\def\@oddhead{\mbox{}\scriptsize\rightmark \hfil \thepage}%
\def\@evenhead{\scriptsize\thepage \hfil \leftmark\mbox{}}%
\def\@oddfoot{}
\def\@evenfoot{}}
\makeatother
\usepackage{amsthm}
\usepackage{float}
\usepackage{graphicx}
\usepackage{subcaption}
\usepackage{psfrag}
\usepackage{url}
\usepackage[ruled,vlined,linesnumbered,commentsnumbered]{algorithm2e}
\usepackage{array}
\usepackage{amsmath,amssymb}
\usepackage{color}
\usepackage{epstopdf}
\usepackage{float}
\usepackage[font={small}]{caption}

\newtheorem{theorem}{Theorem}

\newenvironment{definition}[1][Definition:]{\begin{trivlist}
\item[\hskip \labelsep {\bfseries #1}]}{\end{trivlist}}

\begin{document}
\title{Finding $K$ Contingency List in Power Networks using a New Model of Dependency}
\author{\IEEEauthorblockN{Joydeep Banerjee\textsuperscript{1}, Anamitra Pal\textsuperscript{2}, Kaustav Basu\textsuperscript{1}, Malhar Padhee\textsuperscript{2}, Arunabha Sen\textsuperscript{1}}
\IEEEauthorblockA{\\
\small \textsuperscript{1}School of Computing, Informatics and Decision System Engineering\\
\small \textsuperscript{2}School of Electrical, Computer and Energy Engineering \\
\small Arizona State University\\
\small Tempe, Arizona, USA\\
\small Email: \{jbanerje, kbasu2, mpadhee@asu.edu, anamitra.pal, kbasu2, asen\}@asu.edu}
}
\maketitle 

\begin{abstract}
Smart grid systems are composed of power and communication network components. The components in either network exhibit complex dependencies on components in its own as well as the other network to drive their functionality. Existing, models fail to capture these complex dependencies. In this paper, we restrict to the dependencies in the power network and propose the  Multi-scale Implicative Interdependency Relation (MIIR) model that address the existing limitations. A formal description of the model along with its working dynamics and a brief validation with respect to the 2011 Southwest blackout are provided. Utilizing the MIIR model, the $K$ Contingency List problem is proposed. For a given time instant, the problem solves for a set of $K$ entities in a power network which when failed at that time instant would cause the maximum number of entities to fail eventually. Owing to the problem being NP-complete we devised a Mixed Integer Program (MIP) to obtain the optimal solution and a polynomial time sub-optimal heuristic. The efficacy of the heuristic with respect to the MIP is compared by using different bus system data. In general, the heuristic is shown to provide near optimal solution at a much faster time than the MIP. 
\end{abstract}

\begin{keywords}
MIIR model, Power Network, Cascading failure, Mixed Integer Program, Heuristic, Contingency List.
\end{keywords}

\section{Introduction}
\label{Intro}
The smart grid system constitutes both Power and Communication entities to execute different tasks and can be modeled as an interdependent Power-Communication network. Consider entities in either network. The electricity generation and power flows are partially controlled by the Supervisory Control and Data Acquisition System (SCADA) through signals received from Remote Terminal Units (RTUs). Meanwhile, every entity involved in sensing, sending and controlling the power grid (i.e. the Phasor Measurement Units (PMUs)) are dependent on power network entities to drive their functions. Owing to these dependencies, failure of some entities in either network may eventually result in a cascading failure causing widespread power blackouts.

In this article, we restrict to the analysis of cascading failure in smart grids considering power network dependencies in isolation. Consider that the control centers are not prone to cyber-attacks/physical failures and have additional power backup (independent of the underlying power grid) to make them operational under all circumstances. Even with such an assumption, studying the power network dependencies would provide potential information regarding the vulnerability and reliability of the smart grid system. It is to be noted that analysis under this assumption does bring into effect the failure of PMUs. PMUs are directly dependent for power on the associated power network entity. Hence failure of a power network entity would make the associated PMU non-operational. A brief description of the underlying analysis is provided. The power network entities can be broadly categorized into -- Generation Bus, Load Bus, Neutral Bus and Transmission Lines. Consider a scenario in which there is an initial trip of a transmission line. This would cause some power to be re-routed through other lines to satisfy the load demands. Occasionally, this may result in some transmission lines to have power flow beyond their rated capacity limits causing them to trip. A series/cascade of such line trips might eventually cause a blackout. Regional power blackouts caused by cascading failure of components in power network have been seen previously in New York (2003) ~\cite{andersson2005causes}, San Diego (2011) ~\cite{sandiego} and India (2012) ~\cite{tang2012analysis}. From a system operator's point of view, understanding the vulnerable components in the power network is critical to making well-judged decisions for prevention of power failure. An abstract model that captures these dependencies and an algorithm that uses this model to identify the critical components fast is beneficial from the operator's point of view. 

For the last few years, there has been a considerable surge of research being done in studying the complex dependencies in different critical infrastructures that cause cascading failures. Existing literature relies mostly on graph based models to represent the dependencies \cite{Bul10}, \cite{Gao11}, \cite{Sha11}, \cite{Ros08}, \cite{Par13}, \cite{Ngu13}, \cite{Zus11}. These models, however, fail to capture the complex dependencies that might exist which limits their applicability to real world systems. Authors in \cite{sen2014identification} bring out the need to address the complex dependencies in an inter-dependent or intra-dependent critical infrastructure(s) which can be explained through the following example. Let $e_x,e_w,e_y,e_z$ (which can be a generator, substation, transmission line etc.) be power network entities. Consider the dependency where the entity $e_x$ is operational if (i) entities $e_w$ and (\emph{logical AND}) $e_y$ are operational, or (\emph{logical OR}) (ii) entity $e_z$ is operational. Models in \cite{Bul10}, \cite{Gao11}, \cite{Sha11}, \cite{Ros08}, \cite{Par13}, \cite{Ngu13}, \cite{Zus11}, \cite{Zha05} fails to capture this kind of dependency. Motivated by these findings and limitations of the existing models, the authors in \cite{sen2014identification} proposed a \emph{Boolean logic based dependency} model termed as \emph{Implicative Interdependency Model} (IIM). For the example stated above, the dependency of $e_x$ on $e_w,e_y,e_z$ can be represented as $e_x \leftarrow e_w \cdot e_y + e_z$. This equation representing the dependency of an entity is termed as \textit{Interdependency Relation} (IDR). 
 
The biggest challenge of the IIM formulation is the identification of the dependency relations between the different entities. In the past, this was done based on information obtained from subject matter experts. However, it has now become clear that this is not a very reliable procedure. The second problem is that IIM operates on Boolean logic, implying that the different entities can only have two values --- $0$ or $1$, representing the state of the entity being operational or non-operational. However, this does not provide information about entities which are operating at near-failure state. For example, consider a generator that is operational but operating at $95 \%$ of its peak generation capacity. The IIM model only provides the state of the generator (operational or non-operational). Hence the operator won't be alerted even though the generator is reaching its peak generation capacity (which might eventually fail due to over generation). Thus, such scenarios limits the applicability of the IIM model.

In order to overcome the limitations of IIM and to extend the application domain of IDRs for long-term planning and short-term operational management, the Multi-scale Implicative Interdependency Relation (MIIR) model is proposed. The MIIR model uses the notion of IDRs with added features to capture the power flow in transmission lines and demand/generation of buses. PMU data is used to generate the dependency equations as well as obtaining power flow and demand/generation values.

Using the MIIR model we study the $K$ contingency list problem in this paper. At a given time $t$ the problem solves for a set of $K$ components in the power network which when made non-operational at time $t$ would cause the maximum number of entities to fail. Additionally, the solution would provide insights into components which are operating at their near capacity limits. Such a solution would have an immediate benefit to a system operator making decisions at that time to prevent large-scale power failures. 

The paper is segregated as follows. In Section \ref{model}, we describe the MIIR model in details and its applicability through a case study of real world power failure. The $K$ contingency list problem and its computational complexity are detailed out in Section \ref{ProbForm}. Owing to the problem being NP-complete, a mixed integer program is provided that solves the problem optimally in Section \ref{Sol}. As it is critical to solving the problem fast, a sub-optimal heuristic with polynomial time complexity is discussed in the same section. A detailed experimental analysis for different bus systems is provided in Section \ref{ExpRes}.

\section{Related Work}
\label{Related}
The Presidential Commission on Critical Infrastructures, appointed by President Clinton in 1996 \cite{clintonexecutiveorder1996}, was one of the earliest studies done on understanding robustness and resiliency of Critical Infrastructure Interdependency. Rinaldi et al.\cite{rinaldi2004modeling}, \cite{rinaldi2001identifying} are among the first group of researchers to study the inter-dependency between critical infrastructures and proposed a complex adaptive system to model the inter-dependencies. A survey of Critical Infrastructure Interdependency modeling is presented in  Pederson et al. in \cite{pederson2006critical} which was undertaken by the U.S and a group of international researchers. As discussed before, there have been significant research to capture these inter-dependencies using graph-based models \cite{Bul10}, \cite{Gao11}, \cite{Sha11}, \cite{Ros08}, \cite{Par13}, \cite{Ngu13}, \cite{Zus11}. In general, the abstractions provided by the graph-based models fail to capture the complex inter-dependencies that might exist.

An event driven co-simulation framework for interdependent power-communication network was presented by Lin et al. in \cite{lin2011power}, \cite{lin2012geco}. The authors in \cite{Zha05} presented a game theoretic model for a multi-layer infrastructure networks using flow equilibrium. Authors in \cite{amin2013security} studied the security of interdependent and identical Networked Control System (NCS). They used a discrete-time stochastic linear system to model each plant with a shared communication network controlling the systems. In \cite{muller2016interfacing} the importance of considering power and communication infrastructures simultaneously was highlighted. Based on the Polish power grid blackout, a study of the impact of overhead lines tripping order on the severity of failure is presented in \cite{pfitzner2011controlled}. Analyzing failure in the smart grid under targeted initial attack was studied in \cite{ruj2014analyzing}. The effect of cyber (communication) and power network dependencies in the smart grid was studied in \cite{falahati2012reliability} for reliability assessments. Recovery of information of the failed entities in a power grid after a failure event was studied in \cite{soltan2015joint}. Topology based vulnerability assessment in power netowkr has been reported in \cite{wu2017study}, \cite{bompard2013analysis}, \cite{bompard2011structural}. As we rely on PMU data of the power network for generating the dependency equations, we investigated the current status of utilizing PMU data in addressing the complex inter-dependency issues. The roles of real-time measurements obtained from IEDs such as PMUs in mitigating failures have been investigated in \cite{pal2013applying}, \cite{wang2015multi}, \cite{wang2016use}, \cite{li2014transient}. However, most of the analysis in these papers does not gives insight about the instantaneous identification of most vulnerable entities in a power network. 
\section{The MIIR Model}
In this paper, the MIIR model is developed based on the power network. It is aimed to have a near accurate abstraction of the power flow dynamics and capture cascading failure propagation in the same. It is to be noted with proper modification the model has the potential to be extended for performing a similar analysis in different inter/intra-dependent critical infrastructure system(s).  
\label{model}
\subsection{Model Variables}
We consider load buses, generator buses, neutral buses (or zero injection buses) and transmission lines as different types of entities. Let $E = \{e_1,e_2,\hdots,e_n\}$ denote the set of entities in the power network. Each entity $e_i \in E$ has three values associated with it ---  (i) a lower bound $e_{i,l}$, (ii) an upper bound $e_{i,u}$, and (iii) the instantaneous power value $e_{i,c,t}$ at time $t$ of the entity. For a transmission line type entity $e_k$, the value of $e_{k,c,t}$ provides the power flow in that line at time $t$. Corresponding, for a load bus $e_m$ and generator bus $e_n$ the values $e_{m,c,t}$ and $e_{n,c,t}$ provides the load demand and generating value at time $t$. For the power transmission and distribution systems, PMU data can be used to get the instantaneous power value $e_{i,c,t}$ of the entity $e_i$. The values of $e_{i,l}$ and $e_{i,u}$ can be easily obtained from the entity rating data. For a given time $t$ the state of the entity is still Boolean (operational or not operational) and is guided by the following two factors --- (a) $e_{i,c,t}$ satisfies the property $e_{i,l} < e_{i,c,t} < e_{i,u}$, (b) the corresponding dependency equation of $e_i$ at time $t$ is satisfied. Hence, if $e_i$ has an IDR $e_i \leftarrow e_j \cdot e_k + e_l$ then for $e_i$ to be operational both the properties --- (a) ($e_i$ and $e_k$) or $e_l$ is operational at $t - \delta$, and (b) $e_{i,t} \le e_{i,c,t} \le e_{i,u}$ has to be satisfied. Here $\delta$ refers to the time within which the effect of failure of an entity is propagated to its dependent entity. So using MIIR model the power network at time $t$ is mathematically represented as $P(E,B,C_t,F)$ where $E$ is the set of of entities, $B$ is a set of tuples $\{e_{i,l}, e_{i,u}\}$ ($\forall e_i \in E$) denoting the power value bound on the entity, $C_t$ consist of instantaneous power value $e_{i,c,t}$ ($\forall e_i \in E$) at time $t$ and $F$ contains the set of dependency equations for the entities in $E$. A similar notation has been used in \cite{bompard2013analysis} but our notation brings out a completely different topological aspect of the power network. 

\subsection{Generating and Obtaining the Sets $B,C_t$ and $F$ for a Power Network}
\label{sec32}
We illustrate our strategy to generate the dependency equations $F$ and the set $C_t$ of a power network $P(E,B,C_t,F)$ at a given time $t$. The PMU data is substituted with simulated data due to its unavailability for different bus systems. The MATPOWER \cite{zimmerman2011matpower} software is used to generate the simulated data. For a given time $t$ and a standard bus system (containing a set of buses and transmission lines $E$), the software use load demand of the bus, the impedance of the transmission lines etc. to solve the power flow. The software produces the voltage of each bus in the system as the output. The software suite also includes a wide range of bus systems along with power ratings of the components for all such systems. We restrict ourselves to analyze the real power flow. Firstly, for a given solution, we formally state the procedure to obtain the tuple values of the set $B$ and instantaneous power value contained in the set $C$ (all values in MW) for generator buses, load buses, neutral buses and transmission lines ---

\begin{itemize}
    \item \textit{Generator Bus: }The real part of the power generated is taken as the value of $e_{a,c,t}$ for a generator bus $e_a \in E$. The upper bound $e_{a,u}$ is set to its real generation capacity (supplied in the MATPOWER suite) and the lower bound $e_{a,l}$ is set to $0$. It is to be noted that some generator buses have load demand. Consider $e_x$ be a generator bus with load demand $d$ units and real instantaneous power generated $e_{x,c,t}$ units. Without the loss of generality, such a bus is split into a generator bus $e_{x1}$ with $0$ load demand (instantaneous power generated $e_{x1,c,t}$) and a load bus $e_{x2}$ with instantaneous load demand $d$ units (instantaneous power generated $0$). A transmission line $e_{x12}$ is constructed that connects $e_{x1}$ to $e_{x2}$ with an instantaneous power flow of $d$ units flowing from $e_{x1}$ to $e_{x2}$.
    
    \item \textit{Load Bus: }The real part of the load demand is taken as instantaneous demand value $e_{b,c,t}$ of a load bus $e_b \in E$. For a load bus $e_b$, both its upper and lower bound is set to the instantaneous demand value $e_{b,c,t}$. Essentially, our assumption is that a load bus does not change its demand value irrespective of any failure.
    
    \item \textit{Neutral Bus: }For a netrual bus $e_d \in E$ the values of $e_{d,l}, e_{d,u}$ and $e_{d,c,t}$ are set to $0$.
    
    \item \textit{Transmission Lines: }For two buses $e_1$ and $e_2$ connected by a transmission line $e_{12}$ the power flowing through the transmission line is calculated as $P_{12} = Real(V_1 * (\frac{V_1 - V-2}{I_{12}})^*)$, where $V_1$ is the voltage at bus $e_1$, $V_2$ is the voltage at bus $e_2$ ($V_1$ and $V_2$ returned by the MATPOWER solver) and $I_{12}$ is the impedance of the transmission line $e_{12}$ (obtained from the supplied bus system file of MatPower). $P_{12}$ is the real component of the power flowing in the transmission line $e_{12}$. The lower bound is set to $0$ and the upper bound is taken as the rated capacity of the transmission line. The instantaneous power value $e_{12,c,t}$ is set to $|P_{12}|$ (absolute value).
\end{itemize}

\begin{figure}[ht]
\centering
\includegraphics[width=0.4\textwidth]{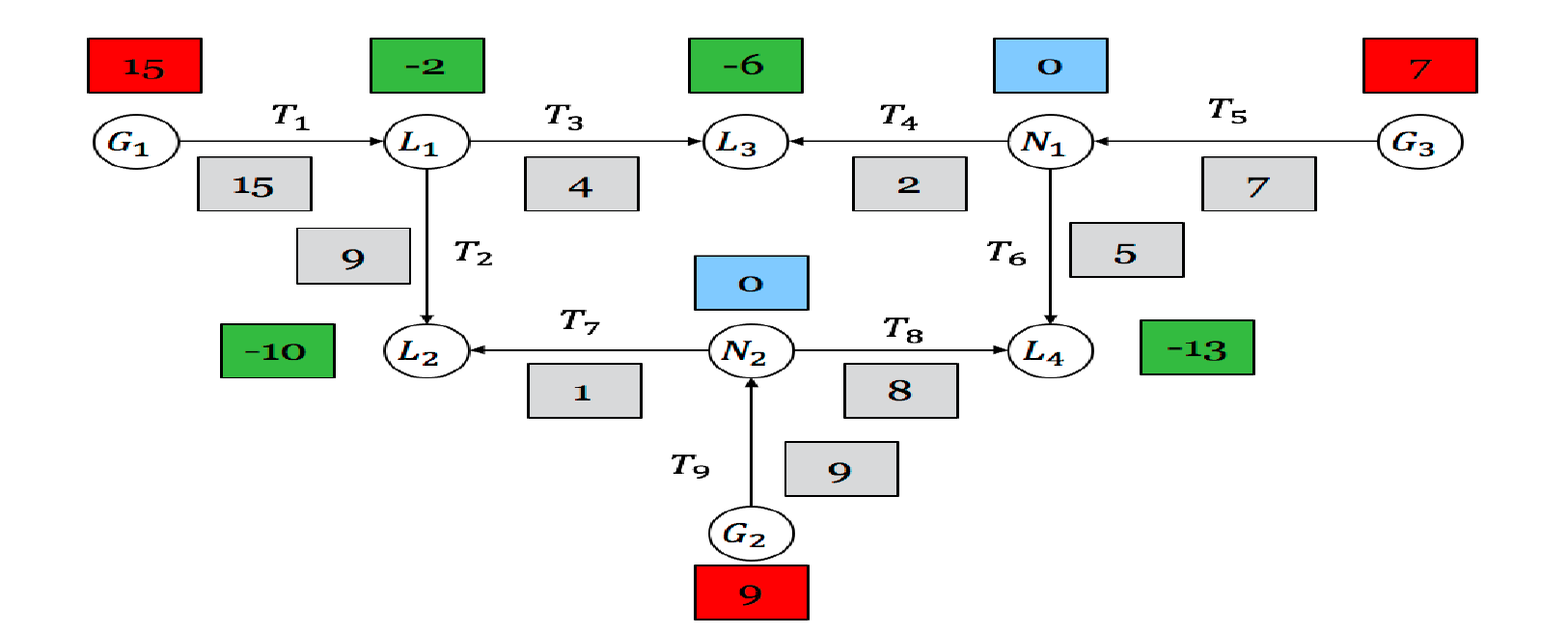}
\caption{A 9 bus power network system with flow, demand, and generation values at an arbitrary time instant $t$}
   \label{fig:dependentPower}
\end{figure}

The description of the transmission line type entity provided above shows that power flows from bus $e_1$ to $e_2$ if $P_{12}$ is positive and vice versa otherwise. As a result, we can interpret the direction of power flow in the line from the solution which we obtained from MatPower. We use this solution to generate the set F which is the set of dependency equations. As an example, we consider the nine bus system which is shown in Figure \ref{fig:dependentPower}. This figure describes a power network $P(E,B,C_t,F)$, at time instance $t$ with $E$ being the set of entities containing generator buses from $G_1$ to $G_3$, load buses $L_1$ to $L_4$, neutral buses $\{N_1,N_2\}$ and transmission lines $T_1$ through $T_9$. The figure also provides the instantaneous power values by solving the power network flow based on demand/generation at some time instant $t$. The red blocks denote the instantaneous real power generated by a generator, the green blocks denote instantaneous real load demands and the blue nodes are neutral. The values in the grey blocks denote the flow of power in the transmission lines with the arrows denoting the direction of power flow. There aren't any IDRs for the transmission lines. The IDRs for bus $b_1$ is created as follows:  (a) let $b_2$, $b_3$ be the buses and $b_{12}$ (between $b_1$ and $b_2$) and $b_{13}$ between ($b_1$ and $b_3$) be the transmission lines for which power flows from these buses to $b_1$, (b)  the dependency equation for the bus $b_1$ is constructed as disjunction of minterms of size $2$ (this consists of the bus from which the power is flowing and the respective transmission line) with each disjunction corresponding to buses from which power is flowing to it. For this example the dependency equation $b_1 \leftarrow b_{12} b_2 + b_{13} b_{3}$ is created. Using this definition, the dependency equations for the buses in Figure \ref{fig:dependentPower} are as follows --- (a) $L_1 \leftarrow T_1 \cdot G_1$ , (b) $L_2 \leftarrow T_2 \cdot L_1 + T_7 \cdot N_2$, (c) $L_3 \leftarrow T_3 \cdot L_1 + T_4 \cdot N_1$, (d) $L_4 \leftarrow T_6 \cdot N_1 + T_8 \cdot N_2$, (e) $N_1 \leftarrow T_5 \cdot G_3$, (f) $N_2 \leftarrow T_9 \cdot G_2$.

\subsection{Dynamics of the MIIR model}
\label{sec33}
To understand the dynamics of cascading failure in power network based on the MIIR model we first create the abstract representation $P(E,B,C_0,F)$ (which is constructed using technique discussed in Section \ref{sec32}) for a power network at time $t=0$. An event of initial failure is assumed to occur at time $t=0$ with failure cascade propagating in unit time steps. For an entity that is operational at time step $t = \tau$ the following equations are required to be satisfied --- 
\begin{equation} 
\sum_{e_m \in O_{e_g}} e_{m,c,\tau} = \sum_{e_n \in I_{e_g}} e_{n,c,\tau} + e_{g,c,\tau}, \forall e_g \in G
\label{eq1}
\end{equation}

\begin{equation}
\sum_{e_m \in O_{e_l}} e_{m,c,\tau}  = \sum_{e_n \in I_{e_l}} e_{n,c,\tau} - e_{l,c,\tau}, \forall e_l \in L
\label{eq2}
\end{equation}

\begin{equation}
\sum_{e_m \in O_{e_k}} e_{m,c,\tau} = \sum_{e_n \in I_{e_k}} e_{n,c,\tau} + e_{k,c,\tau}, \forall e_k \in N
\label{eq3}
\end{equation}

Equations \ref{eq1}-\ref{eq3} dictates the law of conservation of energy for each bus in the system. We assume that there is no delay in time between the power flowing out from a bus and power flowing into it. 


In equation \ref{eq1}, the transmission line entities that flows power out and into the generator bus $e_g$ (where $G \subset E$ contains all generator buses) are represented by sets $I_{e_g}$ and $O_{e_g}$ respectively. Equation \ref{eq2} and \ref{eq3} uses the same notations of transmission line for load bus $e_l$ (where $L \subset E$ contains all load buses) and neutral bus $e_k$ (where $N \subset E$ contains all neutral buses) respectively. In Equation \ref{eq3}, the value of $e_{k,c,t} = 0$ for all time steps and hence it can be simplified as $\sum_{e_m \in I_{e_k}} e_{m,c,t} = \sum_{e_n \in O_{e_k}} e_{n,c,t}$. Additionally, for a generator bus there is no power injected to it through transmission lines. Hence equation \ref{eq1} can be re-written as $\sum_{e_m \in I_{e_g}} e_{m,c,t} =  e_{g,c,t}, \forall e_g \in G$. We use this abstract representation of the power network i.e. $P(E,B,C_0,F)$ at time $t = 0$. Using this the cascading failure process of the power network on an event of initial failure of $E' \subset E$ at time step $t=0$ is detailed out in Algorithm \ref{AlgCas}.

\begin{algorithm}
\small
	\KwData{A power network $P(E,B,C_0,F)$ at time $t=0$ and a set of initially failing entities $E' \subset E$
		}
	\KwResult{A set of failed entities $S$
		}
	\Begin{			
		Initialize $S \leftarrow E'$, $size \leftarrow 0$ \;
		Increment $t \leftarrow t+1$ \;
		\While{$size \ne |S|$}{
		    Set $size \leftarrow |S|$\;
		    Kill all entities whose dependency equations are not satisfied and add them to set $S$\;
		    Adjust power flow values $e_{k,c,t}$ of transmission line entities $e_k$ and generating values $e_{g,c,t}$ of generator buses $e_g$ such that Equations \ref{eq1} -\ref{eq3} are satisfied.\;
		    Kill all entities whose bounds are not satisfied and add them to set $S$ \;
		    Increment $t \leftarrow t+1$ \;
		}
		\Return{$S$} \;
	}		
\caption{Algorithm describing the failure cascading process in power network using MIIR model}
\label{AlgCas}
\end{algorithm}

In Algorithm 1, at every iteration the flow values are adjusted based on Equations \ref{eq1}-\ref{eq3} at line $8$ and entities are made non-operational based on the two conditions mentioned at lines $7$ and $9$. The cascading process continues if a new entities fail in the previous time step (condition $size \ne |S|$). As evident we assume that there is a delay of $1$ time unit for an entity to become non operational if its dependency equations are unsatisfied. All entities whose bound values are not satisfied are made non operational at that time step after power flow calculation. It is to be noted that the dependency equations are generated from a graph which is \textit{directed acyclic}. Owing to this property the cascade reaches a steady (no new entities are non-operational) within $\mathcal{O}(|E|)$ times steps. This can be explained as follows. Consider a single initial failure of an entity. If no entity fails the cascading algorithm would continue till at most $|E| - 1$ time steps since the maximum distance between two nodes in a directed acyclic graph is $|E| - 1$ (considering each edge having a weight $1$). If more than one entity fails then the cascade is expected to stop before $|E|-1$. Hence the number of cascading time steps is strictly upper bounded by $|E|-1$.

Algorithm \ref{AlgCas} assumes that there exist a method to compute the flow value equations to get the instantaneous power values of entities. This is equivalent to computing the AC power flow equations again (using MatPower) which would be time intensive and does not make use of the abstraction created by the MIIR model to make any decision. Moreover, using general graph theoretical algorithmic techniques might result in multiple solutions of instantaneous power values when solving a given set of power flow equations thus resulting in ambiguity. To counteract this, in our abstraction, we use the notion of \textbf{Worst-Case Cascade Propagation} (WCCP) in Algorithm \ref{AlgCas}. Qualitatively, the instantaneous value of power flows and power generator at every time step $t > 0$ of the cascade is set to a value that would cause the maximum number of entities to fail at the end of the cascade. Computation of this power flow values using WCCP is proved to be NP-complete in Section \ref{ProbForm}. We devise a mixed integer porgram to get the optimal solution and a greedy heuristic to get a sub-optimal solution in polynomial time in Section \ref{Sol} that addresses WCCP.

\subsection{Case Study: The 2011 Southwest Blackout}
In this subsection the performance of MIIR with WCCP is tested on a real power system event: the 2011 Southwest Blackout. All data used in this analysis are obtained from \cite{sandiego}. An abstraction of the Southwest Power System is provided in Figure \ref{fig:southwest}. The abbreviations used in Figure \ref{fig:southwest} are --- Western Electricity Coordinating Council (WECC), Serrano (SE), Devers (DE), San Onofre Nuclear Generating Station (SONGS), San Diego Gas \& Electric (SDG \& E), Miguel (MI), Imperial Valley (IV), Imperial Irrigation District (IID), Comisión Federal de Electricidad's (CFE, corresponding to Baja California Control Area), North Gila (NG), Hassayampa (HA), Palo Verde (PV), and Western Area Power Administration-Lower Colorado (WAPA). The blue, orange and green blocks in Figure \ref{fig:southwest} represents neutral, load and generator buses respectively. The transmission lines are labeled $T1-T17$ with the arrows indicating the directions of the pre-disturbance power flows. On September 8, 2011, an initial trip of the HA-NG transmission line ($T11$) caused blackout in SDG \& E region. The objective here is to see whether MIIR model with worst case cascade propagation is able to capture the power outage.

\begin{figure}[ht]
\centering
\includegraphics[width=8cm]{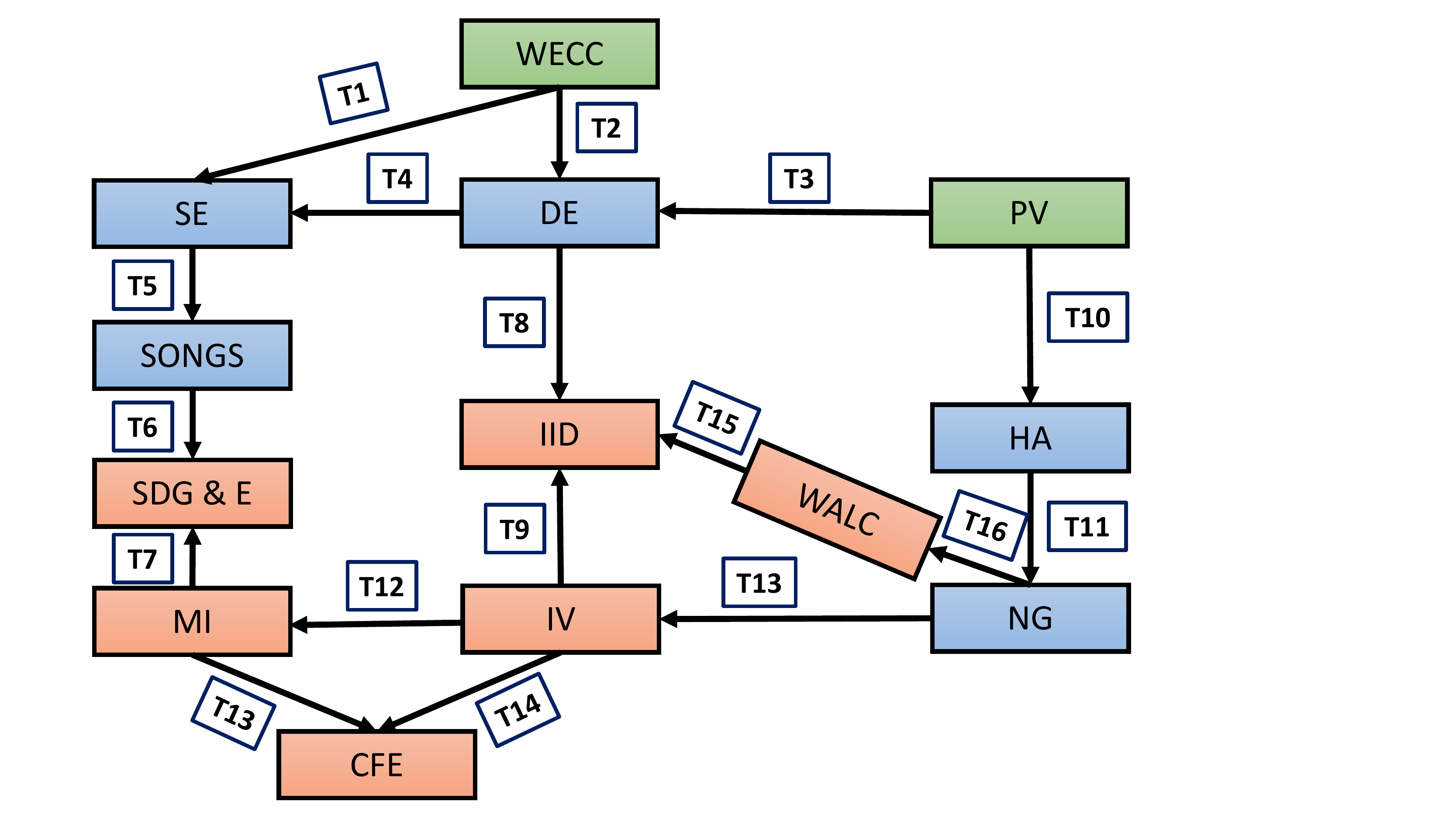}
\caption{An abstraction of the Southwest Power System}
\label{fig:southwest}
\end{figure}

\begin{table}[ht]
\centering
\begin{tabular}{ | l | r | r |} \hline
   {\bf Dependency Equations} \\ \hline
			$SE \leftarrow T1 \cdot WECC + T4 \cdot DE$ \\ \hline
			$DE \leftarrow T3 \cdot PV + T2 \cdot WECC$ \\ \hline
			$SONGS \leftarrow T5 \cdot SE$ \\ \hline 
			$SDG \& E \leftarrow T6 \cdot SONGS + T7 \cdot MI$  \\ \hline
			$IID \leftarrow T8 \cdot DE + T15 \cdot WALC + T9 \cdot IV$ \\ \hline
			$MI \leftarrow T12 \cdot IV$ \\ \hline
			$IV \leftarrow T13 \cdot NG$ \\ \hline 
			$CFE \leftarrow T14 \cdot IV + T13 \cdot MI$ \\ \hline
			$WALC \leftarrow T16 \cdot NG$ \\ \hline
			$NG \leftarrow T11 \cdot HA$ \\ \hline
			$HA \leftarrow T10 \cdot PV$ \\ \hline
	\end{tabular}
	\caption{IDRs of the Southwest Power System}
	\protect\label{tb:depenPowerSouth}
\end{table}

The dependency equations in Table \ref{tb:depenPowerSouth} without the bounds and instantaneous power values of the entities (buses and transmission lines) corresponds to the set $F$. Consider tripping of the entity $T11$ at $t=0$. Just considering the IDRs itself, the component $NG$ fails at $t=1$, $WALC$ and $IV$ at $t=2$, $CFE$ and $MI$ at $t=3$. The pre-disturbance load demands of SDG\&E and IID were approximately 5000 MW and 900 MW, respectively, while the generation bounds on PV and WECC were [0, 4000 MW] and [0, 10000 MW], respectively. After failure of T11, SDG\&E and IID would try to meet their bulk load demands through the generator buses PV and WECC via T6 and T8. The bound on T6 is [0, 2200 MW] and T8 is [0, 1800 MW]. Both PV and WECC have enough generation capacity to meet the load demand of SDG\&E and IID. At $t=3$ owing to the load demand of SDG \& E the transmission line $T6$ would have try to have a power flow of $5000$ MW instantly. Thus $T6$ would trip at $t=3$ causing SDG\&E to trip at $t=4$. Owing to this the power flowing through $T1,T4$ and $T5$ would reduce down to $0$ at $t=4$. The power flow in $T8$ would increase to $900$ MW at $t=3$ for supplying power to IID. Thus the steady state is reached at $t=4$ and MIIR model accurately predicts the blackout of SDG \& E region. 
\section{Problem Formulation}
\label{ProbForm} 
It is important from a power system operator's point of view to understand and know the most critical entities in the network at a given time. This would enable the operator to make more reliable decisions on an event of some failure. For larger systems, an automation that provides the operator with this information would be hugely beneficial. Owing to this we develop the $K$ Contingency List (KCoL) problem using MIIR model with WCCP. For a given time $t$ and an integer $K$ the problem provides the operator with a list of $K$ entities which when failed initially causes the maximum number of entities to fail at the steady state of cascade propagation. Qualitatively, for a given integer $K$ the problem finds a set $E'$ ($|E'| = K$) entities which when failed initially maximizes the total number of entities failed at the end of the cascading process. A formal description of the KCoL problem using WCCP of MIIR model for Power Network is provided --- \\ \\
\textbf{Input:} (a) A power network $P(E,B,C_t,F)$ where $E=G \cup L \cup N \cup T$. Set of entities $G$, $L$, $N$ and $T$ are disjoint and contains the generator buses, load buses, neutral buses and transmission lines respectively. (b) two positive integers $K$ and $S$. \\
\textbf{Decision Version: } Does there exist a set of $K$ entities in $E$ whose failure at time $t$ would result in a failure of at least $S$ entities in total at the end of the cascading process. \\
\textbf{Optimization Version: } Compute the a set $K$ entities in a power network $P(E,B,C_t,F)$ whose failure at time $t$ would maximize the number of entities failed at the steady state of cascade propagation. \\

We prove the problem is NP-complete to solve in Theorem \ref{thm:NPComplete} 

\begin{theorem}
\label{thm:NPComplete}
The KVNE problem using MIIRA model is NP-complete. 
\end{theorem}

\begin{proof}
The problem is proved to NP-complete by a reduction from the densest $p-$subhypergraph problem \cite{hajiaghayi2006minimum}. An instance of the densest $p-$subhypergraph problem consists of a hypergraph $H=(V,E)$ and two parameters $p$ and $M$. The decision version of the problem finds the answer to \textit{whether there exists a set of vertices $V' \subset V$ and $|V'| \le p$ which completely covers at least $M$ hyper-edges}. 

From an instance of the densest $p-$subhypergraph we create an instance of the KCoL problem as follows. We start with an empty set of entities $G$, $L$ and $T$ and an empty set $F$ that would comprise of the dependency equations. A load type entity $L_j$ is added to set $L$ for each hyper-edge $E_j \in E$ with instantaneous load demand $L_{i,c,t}$ set to the number of vertices that comprise this hyper-edge. For each vertex $V_i \in V$ we add a generator type entity $G_i$ to set $G$. The upper bound on the capacity of the generator $G_i$ is set to the sum of all instantaneous load demands $L_{i,c,t} + 1$ for which the corresponding hyper-edge $E_j$ contains the vertex $V_i$. For each hyper-edge $E_j$ consisting of vertices $V_x, V_y,V_z$ (say) three transmission line type entities $T_x, T_y$ and $T_z$ are added to set $T$ and a dependency equation  $L_j \leftarrow T_x \cdot G_x + T_y \cdot G_y + T_z \cdot G_z$ is created and added to set $F$. The upper bound of the transmission line is set to the load demand $+1$ of the entity it connects to (e.g., in this case, the maximum capacity of each transmission line $T_x, T_y, T_z$ are set to the instantaneous load demand $L_{i,c,t} + 1$). The parameter $S$ of KCoL problem is set to $p + M$ and $K$ is set to $p$ (i.e. $p$ entities fail at time $t$). Thus the created instance satisfy the property of the graph from which the dependency relations are computed being \textit{Directed Acyclic}. In the initial operating condition at time $t$, all transmission lines have a line flow value of $1$ unit with each generator $G_i$ producing $P_i$ units of power, where $P_i$ is the number of load entities it is connected to. Hence all load demands are satisfied. It can be directly followed that an instance of KCoL problem can be created from an instance of densest $p-$subhypergraph problem in polynomial time. 

It is to be noted that for the created instance --- (1) Each transmission line has the capacity to satisfy the complete load demand of the load type entity it is connecting, (ii) Each generator has the capacity to satisfy the load demand of all the load type entities it is connected to. Hence an initial failure of one or more entities would not cause any transmission line or generator to trip (fail) because of exceeding its maximum capacity. Thus the generators and transmission lines are susceptible only to initial failure whereas the load entities are vulnerable to both initial and induced failures. However, failure of load entities can not cause any induced failure. Induced failure of the load entity can be caused only when each minterm in its dependency equation have at least one failed entity. Thus no entity fail due to change in power flow values. 

Now consider there exist a solution to the densest $p-$subhypergraph problem. Hence there exist a set of $p$ vertices $V'$ that completely covers $M$ hyper-edges. Failing the generator type entities corresponding to the vertices in $V'$ would thus fail at least $M$ load entities at $t+1$ according to the instance construction. Thus a total of at least $p + M$ entities would fail which solves the KCoL problem. On the other way round consider there exist a solution to the KCoL problem. As reasoned earlier, a load entity cannot cause any induced failure. Hence if a load entity is in the solution then it can be substituted with any operational generator entity without loss of correctness. Similarly, if a transmission line type entity is in the solution it can be replaced by a generator type entity it is connected to. Using this substitution a solution thus comprises of entities $G' \subset G$. All $M$ (or greater than $M$) entities that fail due to the initial failure of $p$ entities belongs to set $L$. Thus the substituted solution (or original solution if no substitution is required) would consist of generator type entities that cause failure of these $M$ (or greater than $M$) load entities. Hence selecting the vertices corresponding to $G'$ would ensure that at least $M$ hyper-edges are completely covered solving the densest $p-$subhypergraph problem. Hence proved.

\end{proof}
\section{Solutions to the Problem}
\label{Sol}
Owing to the KCoL problem being NP-complete, we obtain the optimal solution using an Mixed Integer Program (MIP). However, as we require to compute the contingency list fast, we devise a polynomial time heuristic that provides a sub-optimal solution to the problem. 

\subsection{Optimal Solution using Mixed Integer Program}
As a reference frame, we consider that the initial failure occurs at time step $t = 0$. It is shown in Section \ref{sec33} that the number of time steps in the cascade is upper bounded by $|E|-1$. We devise a Mixed Integer Program that solves the KCoL problem optimally for a power network $P(E,B,C_0,F)$ (the abstraction constructed for $t=0$). Irrespective of whether the steady state is reached before or at time step $|E|-1$, in our MIP we try to maximize the number of entities failed at $t = |E|-1$ when $K$ entities fail at $t = 0$. Moreover, it can not be predicted when the cascading failure stops. Hence, the MIP is bound to check for solution to compute the maximum number of entities that can fail till the maximum possible time step, i.e. $|E|-1$. Firstly, the list of variables used in the MIP formulation are discussed--- 
\begin{itemize}
    \item \textit{Variable List 1: } For each entity $e_i \in E$ a variable set $x_{i,t}, \forall t, 0 \le t \le |E|-1$ are created. The value of $x_{i,t}$ is $0$ if the entity is operational at time step $t$ and $1$ otherwise.  
    
    \item \textit{Variable List 2: } For each entity $e_i \in E$ a variable set $y_{i,t}, \forall t, 0 \le t \le |E|-1$ is created. From the set $C_{0}$ we can get the initial instantaneous power value $e_{i,c,0}$ of an entity $e_i$. The value of $y_{i,0}$ is set to $e_{i,c,0}$. All the instantaneous values are real thus comprising the set of non integer variables in the program.  
\end{itemize}

Using these definitions and the list of variables created, the objective of the MIP is provided in \ref{eqn:MIPobj} and the constraints of the MIP are formally described. 

\begin{equation}\label{eqn:MIPobj}
max\overset{|E|-1}{\underset{i=1}{\sum}}x_{i|E^2|}
\end{equation}

\noindent
\textit{Subjected to:} \\ 

\noindent
{\em Constraint Set 1:} $\sum_{i=1}^{|E|} x_{i,0} = K$. This constraint sets the number of entities failed at time step $t = 0$ to $K$.\\

\noindent
{\em Constraint Set 2:} $x_{i,d} \geq x_{i,t-1}, \forall t, 1 \leq t \leq |E|-1$. This ensures that an entity that is not operational at time step $t = d$ would remain non-operational in all times step $t > d$. \\

\noindent
{\em Constraint Set 3}: Consider an IDR of form ${e_i} \leftarrow e_a \cdot e_b + e_c \cdot e_d$. To capture the cascading failure process, a set of constraints is developed and described below --- 

\vspace{0.05in}
\noindent
{\em Step 1:} New variables are introduced to represent the minterms. In this example, two new variables $c_{ab}$ and $c_{cd}$ are created to represent the terms $e_a \cdot e_b$ and $e_c \cdot e_d$. This is equivalent of adding two new IDRs $c_{ab} \leftarrow e_a \cdot e_b$ and $c_{cd} \leftarrow e_c \cdot e_d$ with the transformed IDR being ${e_i} \leftarrow c_{ab}+c_{cd}$. 

\vspace{0.05in}
\noindent
{\em Step 2:} A linear constraint is developed for the $c$ type variables to capture the failure propagation. For an IDR $c_{ab} \leftarrow e_a \cdot e_b$, the constraint is represented as $c_{ab,t} \leq x_{a,t-1} + x_{b,t-1} , \forall t, 1 \leq t \leq |E|-1$. This captures the condition that $c_{ab,t}$ is equal to $1$ only if at least one of the entities $e_a$ or $e_b$ is non operational. 

\vspace{0.05in}
\noindent
{\em Step 3:} For each transformed IDR a linear constraint is introduced. For an IDR ${e_i} \leftarrow c_{ab} + c_{cd}$ the constraint is represented as $N \times x_{i,t} \leq c_{ab,t-1)} + c_{cd,t-1}, \forall t, 1 \leq t \leq |E| - 1$. Here $N$ is the number of minterms in the IDR (in this example $N=2$). \\

\noindent
{\em Constraint Set 4:} For a given load bus entity $e_l$, the constraint $y_{l,t} = 0, \forall t, 0 \le t \le |E|-1$ is added denoting that the instantaneous power demand of all the load bus remain constant at each time step. Similarly, for a given neutral bus entity $e_n$, the constraint $y_{n,t} = 0, \forall t, 0 \le t \le |E|-1$ is added.\\

\noindent
{\em Constraint Set 5:} For a given generator bus entity $e_p$ and transmission line entity $e_q$, the constraints $x_{p,t} \le \frac{y_{p,t+1}}{e_{p,u}}, \forall t, 1 \le t \le |E|-1$ and $x_{q,t} \le \frac{y_{q,t+1}}{e_{q,u}}, \forall t, 1 \le t \le |E|-1$ are added. As this is a maximization problem, the $x$ type variable of the corresponding generator/transmission line entity would be set to $1$ when it operates beyond its rated upper bound. The constraints $y_{a,t} \ge 0$ and $y_{a,t} \le e_{a,u} + 1$ are added at all time steps for each generator or transmission line type entity $e_a$. This limits the maximum value of these entities to its upper bound plus one and them failing only if their instantaneous power value is just above the upper bound. \\

\noindent
{\em Constraint Set 6:}  To capture the power flow equations in equations \ref{eq1}-\ref{eq3} the following constraints are developed. Consider the equation $\sum_{e_m \in O_{e_l}} e_{m,c,t}  = \sum_{e_n \in I_{e_l}} e_{n,c,t} - e_{l,c,t}$. Naively, this can be constructed as a non-linear constraint $\sum_{e_m \in O_{e_l}} (1-x_{m,t}) \times y_{m,t}  = \sum_{e_n \in I_{e_l}} (1-x_{n,t}) \times y_{n,t} - (1-x_{l,t}) y_{l,t+1}$. The constraint denotes that the instantaneous flow values of the transmission lines/load demand are taken into consideration if the the load bus is operational at the next time step (as failure due to IDR is reflected after $1$ unit of time). This constraint can be linearized as $\sum_{e_m \in O_{e_l}} (y_{m,t} - x_{m,t} \times e_{m,u})  = \sum_{e_n \in I_{e_l}} (y_{n,t} - x_{n,t} \times e_{n,u}) - (y_{l,t} - x_{l,t+1} e_{l,u})$. If a transmission line $e_n$ fails at time instant $t$ then its instantaneous power value is set to its upper bound (owing to constraint set $5$). This would equate the term $(y_{n,t} - x_{n,t} \times e_{n,u})$ corresponding to this transmission line to $0$. If the transmission line $e_n$ is operational then $x_{n,t} = 0$ and hence $(y_{n,t} - x_{n,t} \times e_{n,u})$ would equate to $y_{n,t-1}$ thus being considered in the power flow equation. Similarly if the load bus is not operational the value of $(y_{l,t} - x_{l,t+1} e_{l,u})$ is set to $0$. These constraints are constructed for all time steps $0 \le t \le |E|-2$ and similar constraints are generated for equations \ref{eq1} and \ref{eq3} as well. \\

\noindent
{\em Constraint Set 7:}For each transmission line type entity $e_a \in E$ flowing out power from a bus type entity $e_b$ the constraint $x_{a,t} \le x_{b,t}$ is added for each time step $1 \le t \le |E|-1$. This captures the condition that if a bus type entity fails then all transmission lines to which it transmits power also fails. \\

It is to be noted that there won't be any infeasibility in solution arising due to the constraints. The load and neutral buses can only be made non-operational through their dependency equations. Whereas, the transmission lines and generators can be only made non-operational through change in power flow/generation values (as they don't have any dependency equations). The objective in (\ref{eqn:MIPobj}) along with these set of constraints, finds the the set of $K$ entities whose initial failure at $t=0$ maximizes the number of entities failed at the end of cascading process. As this is a maximization problem the power flow and generation at each time step is set to values that maximizes the total number of entities failing at the steady state. Hence the MIP captures the notion of WCCP. 

\subsection{Heuristic Solution}
\label{sec2}
In this section we design a sub-optimal heuristic that finds a solution to the KCoL problem in polynomial time. Primarily we use the following two definitions that drive the heuristic. 

\begin{table*}[ht]
    \centering
		\begin{tabular}{|c|c|c|c||c|c||c|c||c|c||c|c|}  \hline
			\multicolumn{1}{|c|}{} &
			\multicolumn{1}{|c|}{} &
			\multicolumn{10}{c|}{\bf{$K$ values and number of entities dead in MIP and heuristic}}\\ \cline{3-12}
			\multicolumn{1}{|c}{\bf{DataSet}} &
			\multicolumn{1}{|c}{\bf{Number of}} &
			\multicolumn{2}{|c||}{\bf{1}} & \multicolumn{2}{|c||}{\bf{2}} & \multicolumn{2}{|c||}{\bf{3}} & \multicolumn{2}{|c||}{\bf{4}} & \multicolumn{2}{|c|}{\bf{5}}\\ \cline{3-12}
			\cline{3-12} & \bf{Entities} & \bf{MIP} & \bf{Heu} & \bf{MIP} & \bf{Heu} & \bf{MIP} & \bf{Heu} & \bf{MIP} & \bf{Heu} & \bf{MIP} & \bf{Heu} \\ \hline
			\bf{9 bus} & $24$ & $15$ & $15$ & $19$ & $19$ & $24$ & $24$ & $24$ & $24$ & $24$ & $24$\\ \hline
			\bf{14 bus} & $44$ & $29$ & $29$ & $32$ & $31$ & $36$ & $35$ & $42$ & $41$ & $44$ & $43$\\ \hline
		    \bf{24 bus} & $80$ & $48$ & $48$ & $51$ & $51$ & $54$ & $54$ & $58$ & $58$ & $61$ & $61$\\ \hline
		    \bf{30 bus} & $83$ & $54$ & $54$ & $58$ & $58$ & $64$ & $63$ & $71$ & $70$ & $74$ & $73$\\ \hline
			\bf{39 bus} & $105$ & $63$ & $63$ & $70$ & $70$ & $74$ & $74$ & $79$ & $78$ & $84$ & $82$\\ \hline	
			\bf{57 bus} & $149$ & $96$ & $96$ & $113$ & $108$ & $128$ & $119$ & $136$ & $128$ & $140$ & $136$\\ \hline
			\bf{118 bus} & $405$ & $240$ & $240$ & $245$ & $245$ & $250$ & $250$ & $255$ & $254$ & $260$ & $260$\\ \hline
			\bf{145 bus} & $666$ & $499$ & $499$ & $511$ & $511$ & $517$ & $516$ & $522$ & $517$ & $527$ & $523$\\ \hline
			\bf{300 bus} & $847$ & $493$ & $491$ & $506$ & $506$ & $513$ & $513$ & $524$ & $524$ & $531$ & $530$\\ \hline
			\bf{2383wp bus} & $5923$ & $3249$ & $3249$ & $3296$ & $3278$ & $3326$ & $3308$ & $3354$ & $3332$ & $3382$ & $3353$\\ \hline
		\end{tabular}
		\captionsetup{justification=centering}
		\caption{Quality of solution comparison of Mixed Integer Program and Heuristic for different Power Network bus systems and varying $K$}
		\protect\label{tbl:Quality}
\end{table*}

\begin{table*}[ht]
    \centering
		\begin{tabular}{|c|c|c|c||c|c||c|c||c|c||c|c|}  \hline
			\multicolumn{1}{|c|}{} &
			\multicolumn{1}{|c|}{} &
			\multicolumn{10}{c|}{\bf{$K$ values and running time in sec for MIP and in ms for Heuristic}}\\ \cline{3-12}
			\multicolumn{1}{|c}{\bf{DataSet}} &
			\multicolumn{1}{|c}{\bf{IDR Generation}} &
			\multicolumn{2}{|c||}{\bf{1}} & \multicolumn{2}{|c||}{\bf{2}} & \multicolumn{2}{|c||}{\bf{3}} & \multicolumn{2}{|c||}{\bf{4}} & \multicolumn{2}{|c|}{\bf{5}}\\ \cline{3-12}
			\cline{3-12} & \bf{Time (ms)} & \bf{MIP} & \bf{Heu} & \bf{MIP} & \bf{Heu} & \bf{MIP} & \bf{Heu} & \bf{MIP} & \bf{Heu} & \bf{MIP} & \bf{Heu} \\ \hline
			\bf{9 bus} & $0.31$ & $0.026$ & $2$ & $0.13$ & $1$ & $0.02$ & $1$ & $0.02$ & $1$ & $0.01$ & $1$\\ \hline
			\bf{14 bus} & $0.34$ & $0.09$ & $2$ & $0.14$ & $1$ & $0.08$ & $1$ & $0.05$ & $1$ & $0.04$ & $1$\\ \hline
		    \bf{24 bus} & $0.60$ & $0.27$ & $4$ & $0.16$ & $4$ & $0.19$ & $6$ & $0.18$ & $3$ & $0.27$ & $3$\\ \hline
		    \bf{30 bus} & $0.55$ & $0.24$ & $5$ & $0.20$ & $5$ & $0.17$ & $6$ & $0.16$ & $3$ & $0.18$ & $4$\\ \hline
			\bf{39 bus} & $0.62$ & $0.30$ & $8$ & $0.17$ & $14$ & $0.28$ & $6$ & $0.28$ & $7$ & $0.25$ & $5$\\ \hline	
			\bf{57 bus} & $0.76$ & $1.44$ & $17$ & $0.65$ & $13$ & $0.60$ & $15$ & $0.74$ & $15$ & $0.60$ & $17$\\ \hline
			\bf{118 bus} & $1.49$ & $1.78$ & $33$ & $1.09$ & $52$ & $1.16$ & $64$ & $0.73$ & $33$ & $1.02$ & $41$\\ \hline
			\bf{145 bus} & $2.22$ & $10.48$ & $60$ & $9.45$ & $83$ & $9.49$ & $91$ & $10.26$ & $147$ & $12.72$ & $76$\\ \hline
			\bf{300 bus} & $2.09$ & $11.09$ & $146$ & $4.23$ & $249$ & $4.68$ & $364$ & $3.63$ & $253$ & $4.34$ & $293$\\ \hline
			\bf{2383wp bus} & $7.66$ & $315$ & $5$ & $220$ & $7$ & $225$ & $11$ & $218$ & $14$ & $262$ & $18$\\ \hline
		\end{tabular}
		\captionsetup{justification=centering}
		\caption{Run time comparison of Mixed Integer Program and Heuristic for different Power Network bus systems and varying $K$}
		\protect\label{tbl:Runtime}
\end{table*}

\begin{definition}
	\label{KS}
	\emph{Kill Set:} For an entity $e_j \in E$ in an power network $P(E,B,C_t,F)$, the Kill Set is denoted as $KS(e_j)$. This set comprises of all the entities made non-operational including the entity itself at the end of the cascading process when entity $e_j$ is made non-operational initially.
\end{definition} 

\begin{definition}
	\label{FM}
	\emph{Fractional Minterm Hit Value:} For an entity $e_j \in E$ in an power network $P(E,B,C_t,F)$ the Fractional Minterm Hit Value is denoted as $FMHV(e_j)$. It is calculated as $FMHV(e_j) = \sum_{i = 1}^{m} \frac{c_i}{|s_i|}$ where $m$ are the minterms and for a given minterm $m_i$, $c_i$ are the number of entities that belong to $Kill(e_j)$ and $|s_i|$ is its size. This metric provides an estimate of impact of other operational entities that can be made non-operational at future time steps if the entity $e_j$ is made non-operational.
\end{definition} 

\begin{algorithm}[ht!]
\small
	\KwData{An power network $P(E,C_0,F)$ at time $t=0$ and an integer $K$. 
		}
	\KwResult{A set of initially failing entities $E' \subset E$ and $|E'| \le K$
		}
	\Begin{			
		Initialize $\mathcal{D} \leftarrow \emptyset$\, $E' \leftarrow \emptyset$ and $K_H \leftarrow 0$\;
			\While  {$K_H < K$}{
			     Update $K_H \leftarrow K_H + 1$\;
				 For each entity $e_i \in E \backslash \mathcal{D}$ compute the kill set $C_{e_i}$\;
				 For each entity $e_i \in E\backslash \mathcal{D}$ compute $FMHV(e_i)$\;
				Let $e_j$ be the entity having highest $|C_{e_j}|$ \;
				\If{There exists multiple entities having highest cardinality Kill Set}{
					Let $e_p$ be an entity having highest $FMHV(e_p)$ with $e_p$ in the set of entities having highest cardinality Kill Set\;
					If there is a tie choose arbitrarily\;
					Update $E' \leftarrow E' \cup \{ e_p \}$, $\mathcal{D} \leftarrow \mathcal{D} \cup C_{e_p}$ \;
     				Update all dependencies in $F$ by removing entities in the left and right side of the IDRs that belong to $C_{e_p}$\;
				}
				\Else{
					Update $E' \leftarrow E' \cup \{ e_j \}$, $\mathcal{D} \leftarrow \mathcal{D} \cup C_{x_j}$ \;
     				Update all dependencies in $F$ by removing entities in the left and right side of the IDRs that belong to $C_{e_p}$\;
				}
			}
			\Return{$E'$} \;
	}		
\caption{Heuristic Solution to KoCL problem}
\label{alg:alg2}
\end{algorithm}

Algorithm \ref{alg:alg2} returns a sub-optimal value of $E'$ which when failed initially would greedily maximize (based on \emph{Kill Set} and \emph{FMHV}) the number of entities failed at the end of the cascade. The algorithm runs in $\mathcal{O}(n^3)$ where $n = |E|$. It is to be note that the greedy failure maximization is done based on IDR. To get the actual number of entities failed when the set of entities $E'$ fail initially we use the MIP. Essentially, we modify the constraint 1 such that only entities in $E'$ fail at $t=0$ and see the number of entites failed at the final time step. This gives us a measure to compare the efficacy of the heuristic solution with respect to the MIP. 
\section{Experimental Results}
\label{ExpRes}
We analyzed the run time performance and quality of the heuristic solution with respect to MIP for different bus systems. The quality of the solution is defined by the number of components reported to be non-operational for a given value of $K$. Specifically we used the $9, 14, 24, 30, 39, 57, 118, 145, 300$ and $2383$ Winter Polish bus systems available in MATPOWER. For a given bus system, we used the MATPOWER AC power solver. Using the data, the abstract power network $P(E,C_0,F)$ was generated. On the constructed power network the MIP and heuristic solutions were executed. The implementation was done in Java and a student licensed version of IBM CPLEX optimizer was used to solve the MIP. A UNIX system with 8 GB of RAM and intel i5 processor was used for the execution. All the implementation along with data sets are made open source. The repository can be found in \emph{https://github.com/jbanerje1989/ContingencyList}. 

In Table \ref{tbl:Quality}, a comparison between the MIP and the heuristic solution with respect to the number of entities in non-operational state for different bus systems with $K$ varied from $1$ to $5$ in steps of $1$ are provided. Additionally, the total number of entities (buses and transmission lines) for each bus system is mentioned. It is to be noted that the total number of entities and the number of entities in non-operational state include the entities constructed for generator buses with non-zero load demand (as mentioned in Section \ref{sec32}). Table \ref{tbl:Runtime} reports the IDR generation time for each bus system along with the time taken to execute the MIP and heuristic for different values of $K$. Some insightful observations from the results are as follows --- (a) The heuristic solution performs very nearly to that of MIP with respect to quality and have an almost same performance for $K=1$. (b) For almost all the cases, the maximum percent difference in the number of non operational entities in heuristic with respect to MIP is under $1 \%$ with a maximum percent difference of $7 \%$ for $57$ bus system at $K=3$. (c) It is observed that more than $50 \%$ of the total entities in a given bus system will be non-operational if $K=1$. This implies that the power system are extremely vulnerable even if a single entity is attacked, (d) For almost all the bus systems from $9$ to $300$ the heuristic finds a solution to the \emph{KCol} problem nearly $100$ times faster than the MIP, (e) for the $2383$ Winter Polish bus system the heuristic is $10$ to $20$ times faster. However, it is be noted that the comparison is done based on a serialized implementation of the heuristic. There exists scope to parallelize the heuristic to achieve even faster run time.

Hence, it can be reasonably argued that the Heuristic solves the \emph{KCol} problem achieving near optimal solution at a much faster time compared to MIP. Thus the abstraction provided by the MIIR model along with the Heuristic solution can be used by a Power Network operator to obtain the \emph{$K$ Contingency List} at any given time and make appropriate control decisions.

\section{Conclusion}
\label{Conc}
In this paper, we introduce the MIIR model to capture dependencies and analysis of cascading failure in power network. We used the model to solve the $K$ Contingency list problem that finds a set of $K$ entities which when made non-operational at a certain time step would cause the maximum number of entities to fail at the end of cascading process. As the problem is NP-complete we devise a MIP to obtain the optimal solution and a polynomial time sub-optimal heuristic. The time and quality performance of the heuristic with respect to the MIP are compared for different bus systems and it is shown that the heuristic outperforms the MIP with respect to running time and additionally produces near optimal solution.
\begin{footnotesize}
\bibliographystyle{IEEEtran}
\bibliography{references,referencesBibToAdd}
\end{footnotesize}
\end{document}